\documentclass[11pt]{article}

\usepackage[margin=1in]{geometry}
\usepackage[utf8]{inputenc}
\usepackage{microtype}
\usepackage{amsmath,amssymb,amsthm,mathtools}
\usepackage{xcolor}
\usepackage[colorlinks=true,linkcolor=blue,citecolor=blue,urlcolor=blue]{hyperref}
\usepackage{booktabs}
\usepackage{enumitem}

\newtheorem{theorem}{Theorem}[section]
\newtheorem{lemma}[theorem]{Lemma}
\newtheorem{proposition}[theorem]{Proposition}
\newtheorem{corollary}[theorem]{Corollary}
\newtheorem{claim}[theorem]{Claim}
\newtheorem{definition}[theorem]{Definition}

\newtheorem{openproblem}[theorem]{Open Problem}
\newtheorem*{theorem*}{Theorem}

\usepackage[nameinlink,noabbrev]{cleveref}

\crefname{theorem}{Theorem}{Theorems}
\Crefname{theorem}{Theorem}{Theorems}
\crefname{lemma}{Lemma}{Lemmas}
\Crefname{lemma}{Lemma}{Lemmas}
\Crefname{claim}{Claim}{Claims}
\crefname{proposition}{Proposition}{Propositions}
\Crefname{proposition}{Proposition}{Propositions}
\crefname{corollary}{Corollary}{Corollaries}
\Crefname{corollary}{Corollary}{Corollaries}
\crefname{definition}{Definition}{Definitions}
\Crefname{definition}{Definition}{Definitions}
\crefname{remark}{Remark}{Remarks}
\Crefname{remark}{Remark}{Remarks}
\crefname{conjecture}{Conjecture}{Conjectures}
\Crefname{conjecture}{Conjecture}{Conjectures}
\crefname{openproblem}{Open Problem}{Open Problems}
\Crefname{openproblem}{Open Problem}{Open Problems}

\newcommand{\E}{\mathop{\mathbb E}}

\newcommand{\eps}{\varepsilon}
\newcommand{\F}{\mathbb F}
\newcommand{\Z}{\mathbb Z}
\newcommand{\ML}{\operatorname{ML}}

\newcommand{\ms}{\operatorname{MS}}

\title{Linear Hashing is Not That Awesome}
\author{Or Zamir \\ Tel Aviv University}
\date{}

\begin{document}
\maketitle

\begin{abstract}
Consider the canonical universal hash family
\(
        h(x)= ((ax+b)\bmod p)\bmod m,
\)
where~$a,b$ are chosen uniformly from~$\mathbb Z_p$, which we call linear hashing,
being used to hash~$n$ elements into~$m=\Theta(n)$ buckets.
For any universal family,  the expected size of the largest bucket is at least~$\Omega(\log n / \log\log n)$ and at most~$O(\sqrt{n})$.
The only improvement upon these trivial bounds for linear hashing is a 2019 upper bound of~$\tilde{O}(n^{1/3})$ by Knudsen.
We show that for any~$p$ sufficiently larger than~$n$, there is a set of~$n$ keys whose expected maximum load is
\( n^{\Omega(1/\log\log n)},
\)
proving linear hashing does not have a polylogarithmic maximum load.
We extend the same bounds to the classical multiply-shift hash family of Dietzfelbinger, Hagerup, Katajainen, and Penttonen.  

Our main contribution is an equivalence between the maximum load problem to a density variant of arithmetic Kakeya sets. 
We then complete the lower bound using a construction of Green and Ruzsa of a small set containing long arithmetic progressions with every difference in a prescribed range.
Surprisingly, our equivalence also implies that any substantial improvement over Knudsen's upper bound would imply new results about standard arithmetic Kakeya sets. 
More precisely, an \(O(n^{1/3-\varepsilon})\) upper bound would improve known bounds for unions of complete integer arithmetic progressions, while an \(n^{o(1)}\) upper bound would imply Bourgain's arithmetic-progression criterion.
\end{abstract}

\begingroup
\renewcommand\thefootnote{}
\footnotetext{Independent concurrent work of~\cite{bakshi2026lowerboundslinearhashing} obtains the same super-polylogarithmic lower bound for linear hashing. The two works differ in their surrounding frameworks and additional results.}
\endgroup

\newpage
\tableofcontents
\newpage

\section{Introduction}

Consider the textbook hash family, denoted \emph{linear hashing}, that maps a key $x$ from a universe~$[p]$ to
\[
        h_{a,b}(x)=((ax+b)\bmod p)\bmod m,
\]
where $p$ is prime and $a,b$ are uniform in $\F_p$.  
This is the canonical example of a universal family in the sense of Carter and Wegman~\cite{CarterWegman1979}. 
Assume we map~$n$ distinct keys to~$m=\Theta(n)$ buckets using~$h_{a,b}$.
While universality guarantees that the expected number of elements colliding with each specific element is constant, the expected size of the largest bucket is far less understood.

For a truly random function, this expected \emph{maximum load} is $\Theta(\log n/\log\log n)$.  
The same order is unavoidable for any hash family when~$p$ is sufficiently larger than~$n$: We can draw~$n$ elements to insert independently, and observe the above lower bound holds for~$n$ i.i.d. placements from \emph{any} distribution.
For any \emph{universal} family, the expected maximum load is at most~$O(\sqrt{n})$, as the expected number of element pairs with colliding hashes is~$O(n^2 \cdot \frac{1}{m})=O(n)$ due to universality and linearity of expectation, as well as the number of pair collisions being at least quadratic in the maximum load.

The only known improvement over these straightforward bounds came from Knudsen~\cite{Knudsen2019}, who proved that for linear hashing the expected maximum load is at most~$\tilde{O}\left(n^{1/3}\right)$.
Until this work, no improvements over the naive lower bound were known.
Closing this gap has been repeatedly posed as a significant open problem~\cite{Knudsen2019, Westover2024, dhar2024linear}, with several papers posing the question of whether the upper bound can be improved to match the simple lower bound. 

In this paper we give the first non-trivial lower bound, showing that the expected maximum load for linear hashing is at least super-polylogarithmic.  More precisely, we prove the following.

\begin{theorem*}[Main lower bound, informal]
For every sufficiently large $n$ and~$m=\Theta(n)$, and every prime~$p\geq  n^{1+1/\log\log n}$, there is a set of $n$ keys in $\Z_p$ for which the expected maximum
load of linear hashing is at least
\[
        \exp\left(
          \left(\frac{\log 2}{3}-o(1)\right)
          \frac{\log n}{\log\log n}
        \right)=n^{\Omega\left(\frac{1}{\log\log n}\right)}.
\]
\end{theorem*}

In particular, we rule out linear hashing achieving optimal expected maximum load.
We generalize our lower bound also to the multiply-shift hash family of Dietzfelbinger, Hagerup, Katajainen, and Penttonen~\cite{DietzfelbingerEtAl1997}, defined as
\[
h_a(x) = (ax \bmod 2^w)>>(w-\ell),
\]
where~$a$ is uniform odd number in~$\Z_{2^w}$ and~$y>>i=\lfloor\frac{y}{2^i}\rfloor$ denotes the shift-right operation.
This is a variation of the linear hashing family which is optimized for computer implementation, as it replaces arbitrary divisions with binary operations that are significantly quicker on most computer architectures. 
This simplicity has made multiply-shift a standard practical choice for integer hashing  \cite{Thorup2015HighSpeedHashing,RichterAlvarezDittrich2015,LuanChang2022}.  
Knudsen's upper bound also extends to this family~\cite{Knudsen2019}.

\subsection{Connection to Kakeya sets}

Our main contribution is an equivalence between the maximum load of linear hashing and a new density version of arithmetic Kakeya.  
Informally, we seek a small set~$S$ in~$\Z_p$ such that, for many differences~$d$, the set~$S$ contains \emph{many points} from some short arithmetic progression of difference~$d$.  
Unlike in the usual arithmetic Kakeya problem, we do not require the entire progression to be present: the selected points may form an arbitrary subset, which may depend on~$d$. In particular, they may even not contain any actual arithmetic progression at all.  
We refer to such sets as \emph{density-Kakeya sets}.

This rich-subset relaxation is reminiscent of finite-field Furstenberg sets, where one works over a vector space and asks for large intersections with affine subspaces rather than requiring the entire subspaces to be present
\cite{DharDvirLund2021,dhar2024linear}.
Our setting, however, differs both in its ambient algebraic structure and in the bounded arithmetic patterns with respect to which richness is measured.

Beyond our headline lower bound, the equivalence is also consequential in the upper-bound direction. 
As formalized in \Cref{sec:kakeya-comparison}, any polynomial improvement over Knudsen's cube-root exponent would yield new lower bounds for unions of complete arithmetic progressions over the integers. 
At the qualitative endpoint, an \(n^{o(1)}\) maximum-load upper bound (for all~$p>n$) would imply Bourgain's arithmetic-progression criterion for the Euclidean Kakeya conjecture.
In particular, our equivalence implies that any improvement over the currently known upper bounds for the maximum-load in linear hashing would have direct consequences in arithmetic combinatorics. 

Our lower bound uses a construction of Green and Ruzsa~\cite{GreenRuzsa2019}, who constructed a small set containing a complete arithmetic progression of a prescribed length for every difference in a long interval.  Their construction is therefore a genuine arithmetic Kakeya construction, rather than merely a dense one, and we use their set without modification.

Nevertheless, the density relaxation already plays a meaningful role in our argument.  If a set contains an~$L$-term progression of difference~$d$, then, it also contains every~$r$-th point of an~$rL$-term progression of difference~$d/r$.  Thus every difference supplied by the Green--Ruzsa construction can be duplicated to many of its divisors.  These additional differences generally support only a fraction of the corresponding progression and would not be counted by the ordinary Kakeya formulation, but they remain valid in the density variant and are thus still usable for our hashing application.

We remark that dense arithmetic Kakeya is seemingly a substantially weaker requirement than arithmetic Kakeya.  The subsets associated with different differences need not themselves contain long arithmetic progressions, and their density may tend to zero.  This leaves room for improving our lower bound without improving the best-known genuine arithmetic Kakeya constructions.  We currently do not know a substantial separation between the two notions in the arithmetic setting.
We show that polynomial such separations do occur over vector spaces.  
This suggests that the dense arithmetic problem may similarly admit substantially better constructions, potentially leading to a polynomial lower bound for linear hashing.

\subsection{Alternative hash families}
The question of expected maximum load was also studied for other, related, hash families.
In particular, hash families that are \emph{truly linear}, usually with input and output sets realized as vector spaces of sizes~$p^k$ and~$p^\ell$ for some~$k>\ell$ and prime~$p$, are far better understood.
Despite the confusing terminology, the family we study in this paper---which is usually called ``linear hashing''---is in fact neither linear nor affine: as~$m<p$ is inherently co-prime to the prime~$p$, there is no non-trivial linear function mapping any module of size~$p$ to any module of size~$m$. In our case then, the final $\bmod\;m$ operation breaks linearity.

Truly linear hash families are very attractive mathematically, as linearity makes their analysis significantly cleaner. On the other hand, these variants are far less practical, due to multiplications over non-integers being slower to implement: if the chosen vector space is of high dimension over a small field, then the description of each linear function is a large dense matrix and applying it requires matrix-vector multiplication; otherwise the field is a large prime power, and then scalar multiplication requires extension-field arithmetic.

One endpoint of this range is realizing an input space of size which is a power of two as a vector space~$\F_2^k$. 
Then, the family of linear maps in~$\F_2^k\rightarrow \F_2^\ell$ corresponds to~$\ell$-by-$k$ binary matrices. We can thus define the family
\[
h_A(x) = Ax,
\]
for a uniformly chosen matrix~$A$ from~$\F_2^{\ell\times k}$.
For this family, Alon, Dietzfelbinger, Miltersen, Petrank, and Tardos~\cite{AlonEtAl1999} proved that the expected max load is~$O(\ell \log \ell)$. 
Babka~\cite{Babka2018} subsequently observed that their argument can be optimized to give~$O(\ell)$.
Very recently, it was finally improved by Jaber, Kumar, and Zuckerman~\cite{JaberKumarZuckerman2025} to the optimal~$O\left(\frac{\ell}{\log \ell}\right)$ matching the maximum load order of a fully random function.

At the other endpoint of this range, for a large finite field~$\F_q$ we may consider the family of linear maps in~$\F_q^2\rightarrow \F_q$.
Alon et al.~\cite{AlonEtAl1999} prove that for this family, there is an input set~$S\subseteq \F_q^2$ of size~$|S|=\Theta(q)$ for which the expected maximum load is at least~$\Omega(q^{1/3})$, and even~$\Omega(\sqrt{q})$ if~$q$ is a perfect square.

The same input and output cardinalities can therefore support linear map families with dramatically different behavior.  For example, setting
$q=2^r$, we have
\[
    |\F_q^2|=|\F_2^{2r}|=2^{2r}
    \qquad\text{and}\qquad
    |\F_q|=|\F_2^r|=2^r.
\]
Nevertheless, the family of $\F_q$-linear maps $\F_q^2\rightarrow\F_q$ has polynomial worst-case maximum load, whereas the family of all $\F_2$-linear maps $\F_2^{2r}\rightarrow\F_2^r$ has optimal maximum load.  Algebraically, the former is a highly structured proper subfamily of the latter: every $\F_q$-linear map is also~$\F_2$-linear, but not conversely.  
The multiply-shift family is defined on input and output sets of these same cardinalities, namely
\[
    [2^{2r}]\longrightarrow[2^r],
\]
but respects neither of these linear structures.

These two examples may have provided contradictory intuition as to which of the two known bounds for linear hashing and multiply-shift hashing should be closer to the correct answer.

In a complementary high-load regime, in which the average bucket contains many keys, Dhar and Dvir~\cite{dhar2024linear} studied the same families of genuinely linear maps and obtained a stronger, two-sided guarantee on the load of every bucket, also via studying connections to Kakeya-type problems. 

\paragraph{Organization.}
\Cref{sec:overview} gives an overview of the proof.
\Cref{sec:reduction} introduces dense arithmetic Kakeya sets and proves the
reduction to hashing.  In \Cref{sec:linear-lower-bound} we state the
Green--Ruzsa construction as a black box and prove the lower bound for
linear hashing using it.  \Cref{sec:multiply-shift} gives the
multiply-shift lower bound.  We end with discussion and open problems in \Cref{sec:discussion}.
For completeness, Appendix \Cref{app:green-ruzsa} repeats the short Green--Ruzsa
construction.  We do this because we need to analyze several features of their construction which they do not explicitly spell out.

\paragraph{Acknowledgments within and beyond the realms of men.}
All mathematical content originated from the biologically-generated author and discussions with colleagues, but interaction with an AI model was extremely useful in revealing an oversight, which is next detailed. 
The equivalence with the density variant of arithmetic Kakeya was known to the author for a while, and was previously discussed with several colleagues.
In particular, in earlier discussions with Terence Tao and Cosmin Pohoata, Tao brought up the paper of Green and Ruzsa~\cite{GreenRuzsa2019}, and we explored whether their arguments could be adapted to improve the upper bound for the relaxed Kakeya problem.
After the author later gave this context to an AI tool, it made an important observation: the very same Green--Ruzsa paper also contains a construction, which with essentially no modification can be plugged into the aforementioned equivalence to yield an improved lower bound, rather than the upper bound we pursued. 
The author thanks Tao and Pohoata for the aforementioned discussions, and William Kuszmaul for separate, older, discussions on potential upper bounds for the problem.

\section{Overview}
\label{sec:overview}

The paper is organized around a density relaxation of arithmetic Kakeya. 
Informally, instead of requiring a set~$X\subseteq\F_p$ to contain a complete arithmetic progression of a given difference~$d$, we require it only to contain many points from some short progression of difference~$d$. 
For \(X\subseteq\F_p\), a length parameter~\(L\), and a nonzero difference~\(d\in\F_p^*\), define the richness
\[
        R_{X,L}(d)
        :=
        \max\left\{
          |J|:
          c+dJ\subseteq X
          \text{ for some }c\in\F_p
          \text{ and }J\subseteq[L]
        \right\}.
\]
Thus, an \((L,K,\delta)\) dense arithmetic Kakeya set is an~\(X\) for which \(R_{X,L}(d)\ge K\) for at least~\(\delta p\) differences. The sets of indices~\(J\) may be completely arbitrary; in particular, they need not be intervals or contain even a three-term arithmetic progression.

In \Cref{sec:reduction}, we show that this is precisely the combinatorial object underlying linear hashing. Up to constant factors, the expected maximum load of a set of keys is equal to its average richness over all differences. One direction associates every rich difference with a multiplier producing a heavy bucket. Conversely, the preimage of every bucket under every multiplier is a subset of a short arithmetic progression of the corresponding difference. Consequently, the maximum-load problem for linear hashing is equivalent to bounding the density arithmetic Kakeya parameter.

This equivalence also explains the difficulty of improving the known upper bound. Complete integer progressions are a special case of our density witnesses, and the best known general bounds for them already give only the cube-root exponent. Knudsen's upper bound may therefore be viewed as extending this cube-root bound to arbitrary subsets of bounded cyclic progressions. 
As discussed in \Cref{sec:kakeya-comparison}, any polynomial improvement over this upper bound would improve known bounds already for complete integer progressions, while a subpolynomial upper bound would imply the upper-Minkowski Kakeya conjecture.
On the other hand, an improved lower bound requires only a better construction for the weaker density problem. We give some evidence that this distinction can be substantial by exhibiting a polynomial separation between the two variants over finite vector spaces.

For our lower bound, we use a construction of Green and Ruzsa~\cite{GreenRuzsa2019} as a black box. In the parameter range needed here, their construction gives an integer set~$S$ of size at most~$n$ which contains a complete progression of length
\[
        n^{\Omega(1/\log\log n)}
\]
for every difference in a sufficiently long initial interval. Thus, the Green--Ruzsa set is an ordinary arithmetic Kakeya construction, and we use it without modification.

The density relaxation nevertheless plays a crucial role. If~$S$ contains the progression
\[
        c+d\{0,1,\ldots,K-1\},
\]
then, for every invertible~$r$, the same points can also be written as
\[
        c+(d/r)\{0,r,2r,\ldots,(K-1)r\}.
\]
They are therefore a \(K\)-point subset of a longer progression whose difference is the quotient~$d/r$. In this way, every difference supplied by the Green--Ruzsa construction generates many additional differences. We use this to amplify the original interval of differences to a constant fraction of all differences in~$\F_p$. 
Plugging the resulting dense arithmetic Kakeya set into our equivalence proves the lower bound for linear hashing in \Cref{sec:linear-lower-bound}.

Finally, \Cref{sec:multiply-shift} applies the same Green--Ruzsa set and the same difference-amplification idea to multiply-shift hashing. The only change is geometric: linear hashing identifies residue classes modulo~$m$, whereas multiply-shift buckets are consecutive intervals determined by the high bits. In both cases, the amplified progression is short enough that at least a constant fraction of its selected points fall into one bucket. This gives the same \(n^{\Omega(1/\log\log n)}\) lower bound for both hash families.

\section{From Linear Hashing to Density Arithmetic Kakeya}
\label{sec:reduction}

\subsection{The hashing model}

For an integer $z$ and a positive integer $m$, let
$\langle z\rangle_m\in[m]$ denote the standard representative
of $z$ modulo $m$.  
Let $p$ be prime and let $1\le m<p$.  
The linear hashing family consists of the functions
\[
        h_{a,b}(x) = h^{p,m}_{a,b}(x)
        =\langle ax+b\rangle_p\bmod m,
        \qquad a,b\in\F_p.
\]
For a set $X\subseteq\F_p$, denote the maximum load by
\[
        \ML(h,X)
        :=\max_{r\in [m]}
          \big|\{x\in X\;|\;h(x)=r\}\big|.
\]
For $n\le p$, we are interested in the maximum over all input sets~$X$ of the expected maximum load
\[
        \max_{\substack{X\subseteq\F_p\\ |X|=n}}\;\;
          \E_{a,b}\big[\ML(h^{p,m}_{a,b},X)\big].
\]

\subsection{Definition of density-Kakeya}

The following arithmetic object, which is a relaxation of arithmetic Kakeya sets, naturally comes up in our analysis of linear hashing.

\begin{definition}[Density arithmetic Kakeya set]
\label{def:dense-kakeya}
Let $X\subseteq\F_p$, let $1\le K\le L<p$, and let $\Delta\subseteq\F_p^*$.  We say that $X$ is \emph{$(L,K)$-rich in the differences $\Delta$} if, for every $d\in\Delta$, there are $c_d\in\F_p$ and $J_d\subseteq[L]$ with $|J_d|\ge K$ such that
\[
        c_d+dJ_d\subseteq X.
\]
If $|\Delta|\ge\delta p$, we call $X$ an \emph{$(L,K,\delta)$ dense arithmetic Kakeya set}.
\end{definition}

When $K=L$ and $J_d=[L]$ for every~$d\in \Delta$, this is the standard problem of a set containing a complete $K$-arithmetic-progression for every difference.  
The freedom to choose arbitrary $J_d$ is our density relaxation. 
Nonetheless, this means that a true Kakeya set is in particular also a density-Kakeya set.

We note that $K/L$ can be vanishing.
We also remark that the bounded index interval (that is~$J_d\subseteq[L]$) is essential for our application.

\subsection{Reduction from linear hashing to density-Kakeya}

\begin{theorem}[Density-Kakeya implies high maximum load]
\label{thm:kakeya-reduction}
Let $p$ be prime, let $n\le p$, let $X\subseteq\F_p$ have size at most $n$,
and suppose that $X$ is $(L,K)$-rich in a set of differences $\Delta\subseteq\F_p^*$. 
If
\(
        mL\le p,
\)
then,
\[
        \E_{a,b}\big[\ML(h^{p,m}_{a,b},X)\big]
        \ge \frac{K|\Delta|}{2p}.
\]
\end{theorem}

\begin{proof}
Fix $d\in\Delta$, and write each given dense arithmetic progression subset as
\[
        c_d+dJ_d\subseteq X
        \qquad \text{with} \qquad
        \qquad J_d\subseteq[L]
        \quad \text{and} \quad
        |J_d|\ge K.
\]
For $a_d=md^{-1}$ ($\bmod\; p$) and any $b\in\F_p$, the corresponding keys are sent,
before the final reduction modulo $m$, to
\[
        \langle a_dc_d+b+mj\rangle_p,
        \qquad j\in J_d.
\]
Let $A=\langle a_dc_d+b\rangle_p \in[0,p-1]$.  
Prior to the final~$\bmod\; m$, the values $A+mj$ for $j\in[L]$ lie in an interval of length less than $mL\le p$.  
This interval crosses at most one multiple of $p$. 
Consequently, $J_d$ splits into at most two parts, with respect to that (at most) singular crossing.
On the part before the crossing, every value is congruent to $A$ modulo $m$.  
On the part after the crossing, every value is congruent to $(A-p)$ modulo $m$.  
One of these two parts has size at least $K/2$, implying the load claim.

The map $d\rightarrow md^{-1}$ is injective on $\F_p^*$.  
Thus there are $|\Delta|$ distinct multipliers with maximum load at least $K/2$, for any choice of the additive shift.  
Averaging over $a,b\in\F_p$ proves the stated bound.  
\end{proof}

\begin{corollary}
\label{cor:constant-density-reduction}
If $X$ is an $(L,K,\delta)$ dense arithmetic Kakeya set of size at most $n$ and $mL\le p$, then
\[
        \E_{a,b}\big[\ML(h^{p,m}_{a,b},X)\big]
        \ge \delta K/2.
\]
\end{corollary}

The following crucial observation shows we can benefit from the density relaxation even if the starting set contains complete progressions.
Intuitively, we notice that a length~$k$ arithmetic progression with difference~$d$, is automatically also a subset of the same size~$k$ of a length~$rk$ arithmetic progression with difference $dr^{-1}$ for any~$r$. 
Thus, each arithmetic progression gives many rich differences corresponding to divisors of the original difference.

\begin{lemma}[Difference amplification]
\label{lem:difference-amplification}
Let $D,R,K$ be positive integers with $DR<p$.  
Suppose that $X\subseteq\F_p$ contains a $K$-term arithmetic progression of every integer difference $1\le d\le D$.  
Then $X$ is $(RK,K)$-rich in a set $\Delta\subseteq\F_p^*$ of size $\Omega(DR)$.
\end{lemma}

\begin{proof}
For every pair $1\le d\le D$ and $1\le r\le R$ with $\gcd(d,r)=1$, let
\[
        e_{d,r}=dr^{-1}\pmod p.
\]
If $c_d+d\cdot [K]\subseteq X$, then
\[
        c_d+e_{d,r}\bigl(r[K]\bigr)
        =c_d+d[K]\subseteq X.
\]
Since $r[K]\subseteq[RK]$, this certifies richness in the difference $e_{d,r}$.

We next show that these differences are distinct.
Indeed, $e_{d,r}$ and $e_{d',r'}$ are equal if and only if
\[
        dr'\equiv d'r\pmod p.
\]
Both products are at most $DR<p$, so this is simply equality over the integers.
Hence, taken over the rationals~$\mathbb Q$,
\[
\frac{d}{r} = \frac{d'}{r'}.
\]
Furthermore, the two fractions are reduced, and hence $(d,r)=(d',r')$.  

It is thus left to count the number of pairs~$(d,r)\in [D+1]\times [R+1]$ that are co-prime.
A non-co-prime pair has some common divisor $k\ge2$. 
Hence, by a union bound, the number of such pairs is at most
\[
\sum_{k=2}^{\infty}
\left\lfloor\frac{D}{k}\right\rfloor
\left\lfloor\frac{R}{k}\right\rfloor
\le
DR\sum_{k=2}^{\infty}\frac{1}{k^2}
<
\frac{2}{3} DR.
\]
Thus at least $DR/3=\Omega(DR)$ pairs are coprime.
This gives the claimed $\Omega(DR)$ differences.
\end{proof}

We remark that for simplicity Lemma~\ref{lem:difference-amplification} is phrased with respect to an original set~$X$ that contains \emph{true} arithmetic progressions for \emph{all} differences in some interval. 
On the other hand, both assumptions could be relaxed:
If for every difference the set only contains a rich subset of the arithmetic progression, the proof essentially carries without modification. 
On the other hand, having rich differences for only a subset of the differences may invalidate the co-prime pairs counting argument. The given argument still carries through to arbitrary yet \emph{large enough} subsets of the differences range~$[D]$.

\subsection{Reduction from density-Kakeya to linear hashing}

The preceding reduction has a converse when we fix the naturally arising arithmetic progression length
\[
        L=\left\lceil\frac{p}{m}\right\rceil.
\]
For \(X\subseteq\F_p\) and \(d\in\F_p^*\), define the richness
\[
        R_{X,L}(d)
        :=
        \max\left\{
          |J|\; \big| \;
          c+dJ\subseteq X
          \text{ for some }c\in\F_p
          \text{ and }J\subseteq[L]
        \right\}.
\]
Thus \(X\) is \((L,K)\)-rich in difference \(d\) precisely when \(R_{X,L}(d)\ge K\).

\begin{theorem}[Equivalence of linear hashing maximum load and density-Kakeya]
\label{thm:kakeya-equivalence}
Let \(2\le m<p\), let \(L=\lceil p/m\rceil\), and let \(X\subseteq\F_p\) have size \(n\). 
For every \(a\in\F_p^*\) and \(b\in\F_p\), writing
\[
        d_a=ma^{-1}\pmod p,
\]
we have
\[
        \frac{1}{2}R_{X,L}(d_a)
        \le
        \ML(h^{p,m}_{a,b},X)
        \le
        R_{X,L}(d_a).
\]
Consequently,
\[
        \frac{n}{p}
        +\frac{1}{2p}\sum_{d\in\F_p^*}R_{X,L}(d)
        \le
        \E_{a,b}\left[\ML(h^{p,m}_{a,b},X)\right]
        \le
        \frac{n}{p}
        +\frac{1}{p}\sum_{d\in\F_p^*}R_{X,L}(d).
\]
\end{theorem}

\begin{proof}
Fix \(a\ne0\) and \(b\). For a bucket \(r\in[m]\), let
\[
        I_r=\{j\in\Z_{\ge0}:r+mj<p\}.
\]
Since \(|I_r|\le\lceil p/m\rceil=L\), we have \(I_r\subseteq[L]\). Moreover,
\[
        \{x\in\F_p:h^{p,m}_{a,b}(x)=r\}
        =
        a^{-1}(r-b)+d_a I_r.
\]
Intersecting this identity with \(X\), the keys ending in bucket \(r\) form a subset of a cyclic translate of \(d_a[L]\). 
The load of every bucket is therefore at most \(R_{X,L}(d_a)\), proving the upper bound.

The other direction follows from the proof of \Cref{thm:kakeya-reduction}.
\end{proof}

Thus, up to an absolute factor of two and the negligible contribution of \(a=0\), the expected maximum load is exactly the average directional richness.
We also remark for completeness that in both directions we may consider bounds only concerning a single, uniform, density parameter~$K$:
If at least~$\delta$ fraction of differences have richness at least~$K$, then we get a lower bound of~$\Omega(\delta K)$ for the maximum load.
On the other hand, if the maximum load is~$M$, and hence the average richness is~$\Omega(M)$, then by a standard pigeonholing over dyadic intervals we conclude there exist some~$K,\delta$ such that at least~$\delta$ fraction of differences have richness at least~$K$ and~$\delta K = \Omega(M/\log n)$.

\subsection{Comparison and implications on arithmetic Kakeya bounds}
\label{sec:kakeya-comparison}

In this section, we compare this equivalence with what is known for ordinary arithmetic progressions. 
Suppose that we remove both relaxations we made when moving from arithmetic Kakeya to our density variant: the witnesses again must be complete \(K\)-term progressions, and they lie in the integers rather than in \(\F_p\). 
If an \(n\)-element integer set contains \(K\)-term progressions with \(D\) distinct differences, the state-of-the-art union bound of Gilboa and Pinchasi~\cite[Proposition~4.1]{GilboaPinchasi2014} gives, for every fixed \(\eta>0\),
\[
        n\ge c_\eta D^{1/2-\eta}K.
\]
With the scale of parameters relevant for hashing, \(p=\Theta(nL)=\Omega(nK)\), where the last inequality follows as we also have \(K\le L\). Thus, the above bound implies 
\[
        \delta =\frac{D}{p}
        \lessapprox
        \min\left\{
          1,
          \frac{n}{K^{3}}
        \right\}.
\]
Hence, if~$K\geq n^{1/3}$ we have~$\delta K \lessapprox n/K^2 \leq n^{1/3}$, but also otherwise we have~$\delta K \leq 1\cdot n^{1/3}$.
Strikingly, therefore, even after requiring complete progressions over the non-modular integers, the best presently known general bounds still lead to the same cube-root exponent as the best known upper bound for linear hashing.

Knudsen's argument may now be viewed as an extension of this cube-root bound to both relaxations relevant to hashing. 
By our equivalence, his theorem can be interpreted as showing that the full-progression cube-root phenomenon survives when complete integer progressions are replaced by arbitrary subsets of short cyclic progressions.
Notably, every actual arithmetic Kakeya set is trivially also a density-Kakeya one.
Consequently, a significantly better linear hashing upper bound \(O(n^{1/3-\eps})\) would therefore have independent arithmetic-Kakeya consequences in standard settings. 

To be precise, following Green and Ruzsa~\cite{GreenRuzsa2019}, let \(F_K(D)\) denote the minimum cardinality of a set \(A\subseteq\mathbb Z\) containing a \(K\)-term arithmetic progression with common difference \(d\) for every \(1\le d\le D\). Also define the uniform worst-case maximum load
\[
\mathcal H(n):=\sup_{\substack{p>n\ \mathrm{prime}\\ X\subseteq\F_p,\ |X|\le n}}\E_{a,b}\big[\ML(h^{p,n}_{a,b},X)\big].
\]

\begin{proposition}[Hashing upper bounds imply arithmetic Kakeya bounds]
\label{prop:hashing-implies-kakeya}
For all \(K,D\ge2\), writing \(n=F_K(D)\), we have
\[
\mathcal H(n)\ge\frac14\min\left\{K,\frac{D}{n}\right\}.
\]
Consequently, if \(\mathcal H(n)\le Cn^\alpha\) for all \(n\), where \(\alpha>0\), then
\[
F_K(D)\ge\min\left\{\left(\frac{K}{4C}\right)^{1/\alpha},\left(\frac{D}{4C}\right)^{1/(1+\alpha)}\right\}.
\]
\end{proposition}

\begin{proof}
Let \(A\subseteq\mathbb Z\) have size \(n=F_K(D)\) and contain a \(K\)-term progression of every difference \(1\le d\le D\). Set \(q=\max\{D,nK\}\), and by Bertrand's postulate choose a prime \(p\) with \(q<p<2q\). Let \(X\subseteq\F_p\) be the reduction of \(A\) modulo \(p\). Since \(p>D\) and \(p>K\), every one of the given progressions remains a \(K\)-term progression in \(\F_p\), and the \(D\) differences remain distinct and nonzero. Thus \(X\) has size at most \(n\) and is \((K,K)\)-rich in \(D\) differences. Taking \(m=n\) and \(L=K\), we have \(mL=nK<p\), so \Cref{thm:kakeya-reduction} gives
\[
\mathcal H(n)\ge\frac{KD}{2p}>\frac{KD}{4\max\{D,nK\}}=\frac14\min\left\{K,\frac{D}{n}\right\}.
\]
The second assertion follows because either \(K\le4Cn^\alpha\) or \(D/n\le4Cn^\alpha\).
\end{proof}

Suppose, then, that for some fixed \(0<\eps<1/3\) one could improve Knudsen's bound (for all~$p>n$) to
\[
\mathcal H(n)\le n^{1/3-\eps+o(1)}.
\]
Taking
\[
K=\left\lceil D^{(1-3\eps)/(4-3\eps)}\right\rceil
\]
in \Cref{prop:hashing-implies-kakeya} would give
\[
F_K(D)\ge D^{3/(4-3\eps)-o(1)}.
\]
By contrast, the bound of Gilboa and Pinchasi~\cite{GilboaPinchasi2014} gives at this scale only
\[
F_K(D)\ge D^{1/2+(1-3\eps)/(4-3\eps)-o(1)}.
\]
The improvement in the exponent is
\[
\frac{3}{4-3\eps}-\left(\frac12+\frac{1-3\eps}{4-3\eps}\right)=\frac{9\eps}{8-6\eps}>0.
\]
Thus the exponent \(1/3\) for linear hashing is exactly the threshold at which this reduction begins to improve the best known general bounds for unions of complete integer arithmetic progressions.

We next observe that if there is an upper bound of~$n^{o(1)}$ for linear hashing (for all~$p>n$), then a weaker version of the arithmetic Kakeya conjecture, due to Bourgain, follows.
The \emph{arithmetic} Kakeya conjecture is equivalent to saying that for any fixed~$\varepsilon>0$, there exists a fixed~$k=k(\varepsilon)$ for which~$F_k(N)\geq N^{1-\varepsilon}$ for all sufficiently large~$N$. 
The classical \emph{geometric} Kakeya conjecture is known to follow from it.
Bourgain~\cite{Bourgain1991Dirichlet,Bourgain1993Dirichlet} observed that the geometric Kakeya conjecture already follows from the formally weaker assertion that, for every fixed \(\eta>0\), a set containing \(D^\eta\)-term arithmetic progressions with every difference \(1,\ldots,D\) must have size \(D^{1-o(1)}\); see Green and Ruzsa~\cite[Equation~(1.1)]{GreenRuzsa2019} for this formulation.
Unlike the arithmetic Kakeya conjecture, this condition allows the progression length to grow polynomially with the number of differences. Our equivalence also shows that a \(n^{o(1)}\) upper bound for linear hashing would imply exactly this condition.

\begin{corollary}
\label{cor:subpoly-implies-kakeya}
If \(\mathcal H(n)\le n^{o(1)}\), then for every fixed \(\eta>0\),
\[
F_{\lceil D^\eta\rceil}(D)\ge D^{1-o(1)}.
\]
\end{corollary}

\begin{proof}
Let \(K=\lceil D^\eta\rceil\) and \(n=F_K(D)\). Since \(K\le n\le KD\), we have \(\log n=\Theta_\eta(\log D)\), and hence \(\mathcal H(n)=D^{o(1)}\). By \Cref{prop:hashing-implies-kakeya},
\[
\min\left\{K,\frac{D}{n}\right\}\le4\mathcal H(n)=D^{o(1)}.
\]
Since \(K=D^{\eta+o(1)}\), the first term cannot realize the minimum for sufficiently large \(D\). Therefore \(D/n\le D^{o(1)}\), which is equivalent to \(n\ge D^{1-o(1)}\).
\end{proof}

The conclusion of \Cref{cor:subpoly-implies-kakeya} is precisely Bourgain's condition
\[
\liminf_{D\to\infty}\frac{\log F_{\lceil D^\eta\rceil}(D)}{\log D}\ge1
\qquad\text{for every fixed }\eta>0.
\]

The arithmetic Kakeya conjecture of Katz and Tao~\cite{KatzTao2002}, in the equivalent forms discussed by Green and Ruzsa~\cite{GreenRuzsa2019}, predicts a substantially stronger, sub-polynomial, answer for complete integer progressions, and in particular that such improvements exist for true Kakeya sets.
Both the arithmetic Kakeya conjecture and Bourgain's arithmetic-progression criterion concern complete integer progressions. Neither makes any corresponding prediction for arbitrary vanishing-density subsets of cyclic progressions of the kind we show are equivalent to the hashing bound.

We remark that there is evidence from finite vector spaces that the density relaxation can improve the parameter relevant to hashing by a polynomial factor. 
Consider subsets of \(\F_q^2\) of size \(q\), and regard every affine line as having length \(L=q\). 
If \(q\) is a perfect square, the construction of Alon et al.~\cite{AlonEtAl1999} gives a set \(S\subseteq\F_q^2\) of size \(q\) such that, for every line direction, some affine line in that direction intersects \(S\) in at least \(K=\sqrt q\) points. 
Thus all directions are good, and the average density satisfies
\[
        \delta K
        =
        1\cdot{\sqrt q}
        =
        {\sqrt q}.
\]

By contrast, an arbitrary \(q\)-element subset of \(\F_q^2\) can contain a complete line in at most one direction, resulting in $\delta K=O(1)$.
The density relaxation therefore improves precisely the relevant objective by a factor of \(\Omega(\sqrt q)\). We do not know an analogous separation for bounded cyclic progressions in \(\F_p\), which is the setting required for linear hashing.
Nevertheless, this vector-space example shows that such a separation is possible in principle and supports the possibility of obtaining polynomial hashing lower bounds without improving ordinary arithmetic Kakeya constructions.

\section{Lower Bound for Linear Hashing}
\label{sec:linear-lower-bound}

In this section we prove our main lower bound, for linear hashing.
We do so by plugging a construction of Green and Ruzsa~\cite{GreenRuzsa2019} in the reduction of Section~\ref{sec:reduction}.
We begin by presenting and repackaging the Green-Ruzsa construction, as we need to use it with slightly different parameter choices than those they choose in their paper, as well as to explicitly spell out more properties of their construction.

\subsection{Digesting a construction of Green and Ruzsa}

We use the following parameterized form of the construction in Section~5 of Green and Ruzsa~\cite{GreenRuzsa2019}.  
Let $q_1<q_2<\cdots<q_t$ be the first odd primes and write
\[
        Q_t=\prod_{i=1}^t q_i,
        \qquad
        \rho_t=\prod_{i=1}^t\left(1+\frac1{q_i}\right).
\]

\begin{theorem}[Green--Ruzsa construction, repeated in \Cref{app:green-ruzsa}]
\label{thm:GR-black-box}
For all integers $t,K\ge1$, there is a set $S=S(t,K)\subseteq\{1,2,\ldots,KQ_t\}$ with the following properties.
\begin{enumerate}[label=\textup{(\roman*)}]
    \item For every $d\in\{1,2,\ldots,Q_t-1\}$, the set $S$ contains a $K$-term arithmetic progression of difference $d$.
    \item The cardinality of $S$ satisfies
    \[
        |S|\le K^2Q_t2^{-t}\rho_t.
    \]
\end{enumerate}
\end{theorem}

For the first $t$ odd primes, standard prime number estimates (e.g.,~\cite{MontgomeryVaughan2007}) give
\[
        \log Q_t=(1+o(1))t\log t,
        \qquad \text{and} \qquad
        \rho_t=e^{o(t)},
\]
as~$\log q_t=(1+o(1))\log t$ and~$q_t=\omega(t)$.

We next set the parameters in a way that would be useful for our lower bounds.

\begin{corollary}
\label{cor:GR-hashing-scale}
Fix constants $0<c<C$ and $0<\eta<\log 2/3$.  
For all sufficiently large~$n$, let
\[
        K_0=\exp\left(
          \left(\frac{\log 2}{3}-\eta\right)
          \frac{\log n}{\log\log n}
        \right).
\]
For every integer $m$ satisfying $cn\le m\le Cn$, there are positive integers $K,D$ and an integer set~$S$ such that
\[
        K\ge \frac{K_0}{256},
        \qquad
        16mK\le D\le\frac{mK_0}{4},
\]
the set~$S$ has size at most~$n$, and, for every $1\le d\le D$, it contains a $K$-term arithmetic progression of difference~$d$.
\end{corollary}

\begin{proof}
Put
\[
        K=\left\lfloor\frac{K_0}{128}\right\rfloor,
        \qquad
        D=\lceil16mK\rceil.
\]
For sufficiently large~$n$, we have $K\ge K_0/256$ and $16mK\le D\le32mK\le mK_0/4$.

Let $t$ be the minimal integer so that $Q_t>D$, and take the set~$S(t,K)$ from \Cref{thm:GR-black-box}.  
By minimality, $Q_t\le Dq_t$ as~$\frac{Q_t}{q_t}=Q_{t-1}\leq D$.  
Since $D=n^{1+o(1)}$, the prime number estimates above imply
\[
        t=(1+o(1))\frac{\log n}{\log\log n}.
\]
The size estimate in \Cref{thm:GR-black-box} gives
\begin{align*}
        |S|
        &\le K^2Q_t2^{-t}\rho_t\\
        &\le K^2Dq_t2^{-t}\rho_t\\
        &\le 32mK^3q_t2^{-t}\rho_t.
\end{align*}
As $m=\Theta(n)$ and $q_t\rho_t=e^{o(t)}$, it follows that
\begin{align*}
        \log\frac{|S|}{n}
        &\le 3\log K-t\log 2+o(t)\\
        &\le
        \left(-3\eta+o(1)\right)
        \frac{\log n}{\log\log n}.
\end{align*}
Thus $|S|<n$ for all sufficiently large~$n$.  
Finally, $Q_t>D$, so the progression property in \Cref{thm:GR-black-box} gives the required progressions for every difference at most~$D$.
\end{proof}

\subsection{Hashing lower bound}

We now prove the main lower bound.  

\begin{theorem}[Lower bound for linear hashing]
\label{thm:linear-main}
Fix constants $0<c<C$ and $0<\eta<\log2/3$.  
For all sufficiently large~$n$, let
\[
        K_0=\exp\left(
          \left(\frac{\log 2}{3}-\eta\right)
          \frac{\log n}{\log\log n}
        \right).
\]
For every integer $m$ satisfying $cn\le m\le Cn$ and every prime $p\ge mK_0$, 
there is a set $X\subseteq\F_p$ of size~$n$ such that
\[
        \E_{a,b}\big[\ML(h^{p,m}_{a,b},X)\big]=\Omega(K_0).
\]
\end{theorem}

\begin{proof}
Let $K,D,S$ be supplied by \Cref{cor:GR-hashing-scale}.  
Since $p\ge mK_0$ and $D\le mK_0/4$, we have $p\ge4D$.

Reduce $S$ modulo~$p$, and call the resulting set $X_0\subseteq\F_p$.
Reduction can only decrease its size.  
It preserves a $K$-term cyclic progression of every integer difference $1\le d\le D$: 
we have $D<p$ and $K<p$, so the terms within each such progression remain distinct.

Set
\[
        R=\left\lfloor\frac{p}{2D}\right\rfloor.
\]
It follows that
\[
        \frac p4\le DR\le\frac p2.
\]
By \Cref{lem:difference-amplification}, $X_0$ is $(RK,K)$-rich in $\Omega(DR)=\Omega(p)$ differences.  
Moreover,
\[
        mRK
        \le \frac{mKp}{2D}
        \le \frac p{32}.
\]
Pad $X_0$ with arbitrary elements to a set $X\subseteq\F_p$ of exactly~$n$ keys.  
Theorem~\ref{thm:kakeya-reduction} now gives
\[
        \E_{a,b}\big[\ML(h^{p,m}_{a,b},X)\big]
        =\Omega(K)=\Omega(K_0),
\]
as required.
\end{proof}

Taking $\eta$ to zero sufficiently slowly gives the form stated in the introduction.
We further remark that the proof of \Cref{thm:linear-main} is pointwise in the additive shift~$b$: every good multiplier has load at least $K/2$ for every~$b$.  
Thus the same lower bound holds if $b$ is fixed arbitrarily or drawn from any distribution.
In particular, the lower bound holds also for the variant of linear hashing in which we set~$b=0$ deterministically.

\section{Lower Bound for Multiply-Shift Hashing}
\label{sec:multiply-shift}

We next extend the lower bound to the multiply-shift family of Dietzfelbinger, Hagerup, Katajainen, and Penttonen~\cite{DietzfelbingerEtAl1997}.  
Let $q=2^w$ and $m=2^\ell$, with $w>\ell$, and put
\[
        B=\frac qm=2^{w-\ell}.
\]
For an odd multiplier $a\in\Z_q$, define
\[
        \ms^{q,m}_{a}(x)
        =\left\lfloor
          \frac{\langle ax\rangle_q}{B}
        \right\rfloor.
\]
The multiply-shift family takes a uniform odd~$a$.

The first modular operation is very similar to that of linear hashing.
The final operation though, replacing the $\bmod\; m$ reduction, now instead takes the $\ell$ most significant bits of the $w$-bit product. We call this operation \emph{high-bit projection}.
Its equivalence classes are the $m$ consecutive intervals of length~$B$.
We next adapt the lower bound to that different reduction step.

\begin{claim}[Short cyclic arcs are heavy]
\label{lem:consecutive-multiply-shift}
Let $Y\subseteq\Z_q$ consist of $K$ distinct points contained in a cyclic interval of length at most~$B$.  
Then, at least $K/2$ elements of~$Y$ have the same high-bit projection.
\end{claim}

\begin{proof}
A cyclic interval of length at most $B$ meets at most two of the $B$-element buckets, including when it crosses from $q-1$ to $0$. One of the two buckets contains at least half of its elements.
\end{proof}

\begin{theorem}[Lower bound for multiply-shift]
\label{thm:multiply-shift-main}
Fix constants $0<c<C$ and $0<\eta<\log2/3$.  
For all sufficiently large~$n$, let
\[
        K_0=\exp\left(
          \left(\frac{\log2}{3}-\eta\right)
          \frac{\log n}{\log\log n}
        \right).
\]
For every power of two $m$ satisfying $cn\le m\le Cn$ and every power of two $q\ge mK_0$, 
there is a set $X\subseteq\Z_q$ of size~$n$ such that, 
for a uniform odd multiplier~$a$,
\[
        \E_a\big[\ML(\ms^{q,m}_{a},X)\big]=\Omega(K_0).
\]
\end{theorem}

\begin{proof}
Let $K,D,S$ be supplied by \Cref{cor:GR-hashing-scale}.  
Since $q\ge mK_0$ and $D\le mK_0/4$, we have $q\ge4D$.  
Reduce $S$ modulo~$q$ and pad it to a set $X\subseteq\Z_q$ of size~$n$.  
The progressions of odd difference remain injective after reduction because every odd residue is invertible modulo~$q$ and $K<q$.

Let
\[
        R=\left\lfloor\frac{q}{2D}\right\rfloor.
\]
As before, $q\ge4D$ and $q/4\le DR\le q/2$.  
We next recover a variant of~\cref{lem:difference-amplification} that carries despite not working modulo a prime.
For every coprime pair of odd integers $d\le D$ and $r\le R$, define the \emph{admissible} odd multiplier
\[
        a_{d,r}=rd^{-1}\pmod q.
\]
These multipliers are distinct.  
Equality of $a_{d,r}$ and $a_{d',r'}$ is equivalent to
\[
        rd'\equiv r'd\pmod q.
\]
Both sides are at most $DR<q$, so equality holds over the integers, and uniqueness of reduced fractions gives $(d,r)=(d',r')$.  
A positive proportion of all pairs in $[D]\times[R]$ are both odd and coprime.  
Consequently the construction supplies $\Omega(DR)=\Omega(q)$ distinct good odd multipliers.

Fix one such pair, and let $x_d+d[K]$ be the corresponding progression.
Multiplication by $a_{d,r}$ sends it to the $K$ points
\[
        a_{d,r}x_d+r[K]\pmod q.
\]
They lie in a cyclic interval of length at most $RK$, and
\[
        RK
        \le\frac{Kq}{2D}
        \le\frac{q}{32m}
        <B.
\]
By \Cref{lem:consecutive-multiply-shift}, the maximum load is thus at least $K/2$.  
Since a constant fraction of the $q/2$ odd multipliers are good, averaging gives $\Omega(K)=\Omega(K_0)$.  
\end{proof}

\section{Discussion and Open Problems}
\label{sec:discussion}

We proved that the worst-case expected maximum load of both linear hashing and multiply-shift hashing is super-polylogarithmic. Together with Knudsen's upper bound~\cite{Knudsen2019}, the remaining gap is
\[
        n^{\Omega(1/\log\log n)}
        \qquad\text{versus}\qquad
        \widetilde{O}(n^{1/3}).
\]
In particular, the main remaining qualitative question is whether the maximum load is polynomial in~$n$.

\begin{openproblem}
Is there an absolute constant~$c>0$ such that the worst-case expected maximum load of linear hashing is at least~$n^c$? Does the same hold for multiply-shift hashing?
\end{openproblem}

Equivalently, up to logarithmic factors, does there exist a bounded cyclic density arithmetic Kakeya set for which the product between the fraction of good differences and the richness in each good difference is~$n^{\Omega(1)}$?

It is tempting to try to close the gap from above. Our equivalence reveals, however, that such an improvement would have consequences beyond hashing. As discussed in \Cref{sec:kakeya-comparison}, an upper bound of~$O(n^{1/3-\eps})$ holding throughout the relevant parameter range for linear hashing would also improve the best known general bound for ordinary integer sets containing complete arithmetic progressions with many different differences~\cite{GilboaPinchasi2014}. Thus, even after removing both the density relaxation and the modular setting, improving the cube-root exponent would yield new progress on standard arithmetic Kakeya bounds. 
More strongly, an \(n^{o(1)}\) upper bound would imply Bourgain's arithmetic-progression criterion for the geometric (upper-Minkowski) Kakeya conjecture.
This may help explain why the hashing upper bound has resisted improvement. The other possibility is that the cube-root bound is close to the truth. Our lower bound makes polynomial behavior appear plausible.

The situation for lower bounds is quite different. A better lower bound for hashing requires only a better construction for the density variant. Such a construction need not contain any long arithmetic progressions and therefore need not improve the known constructions for ordinary arithmetic Kakeya. There is consequently no analogous arithmetic-Kakeya barrier to obtaining a polynomial hashing lower bound. The vector-space example in \Cref{sec:kakeya-comparison} gives evidence that this distinction can be substantial.
In particular, it is entirely consistent that the arithmetic Kakeya conjecture is true while the maximum load of linear hashing is~$n^{\Omega(1)}$. These two conclusions would say that complete arithmetic progressions cannot overlap too efficiently, but arbitrary rich subsets of such progressions can. Establishing such a separation in the bounded arithmetic setting is therefore a natural problem on its own.

This separation may also be informative for the ordinary arithmetic Kakeya problem. If the full and density variants indeed have substantially different answers, then any proof of the arithmetic Kakeya conjecture must, at least implicitly, distinguish a complete progression from an arbitrary collection of many points on the same progression. 
Understanding where this distinction first becomes usable may therefore provide both a route to stronger hashing lower bounds and a useful guide for the classical arithmetic Kakeya problem.

\bibliography{linear_hashing_refs}

@article{AlonEtAl1999,
  author  = {Noga Alon and Martin Dietzfelbinger and Peter Bro Miltersen and Erez Petrank and G{\'a}bor Tardos},
  title   = {Linear Hash Functions},
  journal = {Journal of the ACM},
  volume  = {46},
  number  = {5},
  pages   = {667--683},
  year    = {1999},
  doi     = {10.1145/324133.324179}
}

@article{CarterWegman1979,
  author  = {J. Lawrence Carter and Mark N. Wegman},
  title   = {Universal Classes of Hash Functions},
  journal = {Journal of Computer and System Sciences},
  volume  = {18},
  number  = {2},
  pages   = {143--154},
  year    = {1979},
  doi     = {10.1016/0022-0000(79)90044-8}
}

@article{DietzfelbingerEtAl1997,
  author  = {Martin Dietzfelbinger and Torben Hagerup and Jyrki Katajainen and Martti Penttonen},
  title   = {A Reliable Randomized Algorithm for the Closest-Pair Problem},
  journal = {Journal of Algorithms},
  volume  = {25},
  number  = {1},
  pages   = {19--51},
  year    = {1997},
  doi     = {10.1006/jagm.1997.0873}
}

@article{GilboaPinchasi2014,
  author  = {Shoni Gilboa and Rom Pinchasi},
  title   = {On the Union of Arithmetic Progressions},
  journal = {SIAM Journal on Discrete Mathematics},
  volume  = {28},
  number  = {3},
  pages   = {1062--1073},
  year    = {2014},
  doi     = {10.1137/130941122}
}

@article{GreenRuzsa2019,
  author  = {Ben Green and Imre Z. Ruzsa},
  title   = {On the Arithmetic {Kakeya} Conjecture of {Katz} and {Tao}},
  journal = {Periodica Mathematica Hungarica},
  volume  = {78},
  pages   = {135--151},
  year    = {2019},
  doi     = {10.1007/s10998-018-0270-z}
}

@inproceedings{JaberKumarZuckerman2025,
  author    = {Michael Jaber and Vinayak M. Kumar and David Zuckerman},
  title     = {Linear Hashing Is Optimal},
  booktitle = {Proceedings of the 57th Annual ACM Symposium on Theory of Computing},
  pages     = {245--255},
  year      = {2025},
  doi       = {10.1145/3717823.3718208}
}

@article{Knudsen2019,
  author  = {Mathias B{\ae}k Tejs Knudsen},
  title   = {Linear Hashing Is Awesome},
  journal = {SIAM Journal on Computing},
  volume  = {48},
  number  = {2},
  pages   = {736--741},
  year    = {2019},
  doi     = {10.1137/17M1126801}
}

@misc{Westover2024,
  author       = {Alek Westover},
  title        = {On the Relationship Between Several Variants of the Linear Hashing Conjecture},
  year         = {2024},
  howpublished = {arXiv:2307.13016}
}

@article{dhar2024linear,
  title={Linear Hashing with $\ell_\infty$ guarantees and two-sided Kakeya bounds},
  author={Dhar, Manik and Dvir, Zeev},
  journal={TheoretiCS},
  volume={3},
  year={2024},
  publisher={Episciences. org}
}

@article{Thorup2015HighSpeedHashing,
  author        = {Mikkel Thorup},
  title         = {High Speed Hashing for Integers and Strings},
  journal       = {arXiv preprint arXiv:1504.06804},
  year          = {2015},
  eprint        = {1504.06804},
  archivePrefix = {arXiv},
  primaryClass  = {cs.DS},
  doi           = {10.48550/arXiv.1504.06804},
  url           = {https://arxiv.org/abs/1504.06804}
}

@article{RichterAlvarezDittrich2015,
  author  = {Stefan Richter and Victor Alvarez and Jens Dittrich},
  title   = {A Seven-Dimensional Analysis of Hashing Methods and Its
             Implications on Query Processing},
  journal = {Proceedings of the VLDB Endowment},
  volume  = {9},
  number  = {3},
  pages   = {96--107},
  year    = {2015},
  doi     = {10.14778/2850583.2850585},
  url     = {https://www.vldb.org/pvldb/vol9/p96-richter.pdf}
}

@article{LuanChang2022,
  author  = {Hua Luan and Lei Chang},
  title   = {An Experimental Study of Group-By and Aggregation on
             {CPU--GPU} Processors},
  journal = {Journal of Engineering and Applied Science},
  volume  = {69},
  number  = {1},
  pages   = {54},
  year    = {2022},
  doi     = {10.1186/s44147-022-00108-1},
  url     = {https://doi.org/10.1186/s44147-022-00108-1} }

@misc{Babka2018,
  author        = {Martin Babka},
  title         = {A Note on the Size of Largest Bins Using Placement
                   with Linear Transformations},
  year          = {2018},
  eprint        = {1810.04161},
  archivePrefix = {arXiv},
  primaryClass  = {cs.DS},
  doi           = {10.48550/arXiv.1810.04161},
  url           = {https://arxiv.org/abs/1810.04161}
}

@article{DharDvirLund2021,
  author  = {Manik Dhar and Zeev Dvir and Ben Lund},
  title   = {Simple Proofs for Furstenberg Sets over Finite Fields},
  journal = {Discrete Analysis},
  volume  = {2021},
  eid     = {22},
  year    = {2021},
  doi     = {10.19086/da.29067},
  url     = {https://doi.org/10.19086/da.29067}
}

@book{MontgomeryVaughan2007,
  author    = {Hugh L. Montgomery and Robert C. Vaughan},
  title     = {Multiplicative Number Theory I: Classical Theory},
  series    = {Cambridge Studies in Advanced Mathematics},
  volume    = {97},
  publisher = {Cambridge University Press},
  year      = {2007}
}

@article{KatzTao2002,
  author  = {Nets Hawk Katz and Terence Tao},
  title   = {New Bounds for {Kakeya} Problems},
  journal = {Journal d'Analyse Mathématique},
  volume  = {87},
  pages   = {231--263},
  year    = {2002},
  doi     = {10.1007/BF02868476}
}

@incollection{Bourgain1991Dirichlet,
  author    = {Jean Bourgain},
  title     = {Remarks on {Montgomery}'s Conjectures on {Dirichlet} Sums},
  booktitle = {Geometric Aspects of Functional Analysis: Israel Seminar (GAFA) 1989--90},
  editor    = {Joram Lindenstrauss and Vitali D. Milman},
  series    = {Lecture Notes in Mathematics},
  volume    = {1469},
  pages     = {153--165},
  publisher = {Springer},
  address   = {Berlin},
  year      = {1991}
}

@article{Bourgain1993Dirichlet,
  author  = {Jean Bourgain},
  title   = {On the Distribution of {Dirichlet} Sums},
  journal = {Journal d'Analyse Mathématique},
  volume  = {60},
  pages   = {21--32},
  year    = {1993},
  doi     = {10.1007/BF03341964}
}

@misc{bakshi2026lowerboundslinearhashing,
      title={Lower Bounds for Linear Hashing via Arithmetic Kakeya}, 
      author={Ainesh Bakshi and Alex Conway and Hanna Komlós and William Kuszmaul and Alek Westover},
      year={2026},
      eprint={2608.24866},
      archivePrefix={arXiv},
      primaryClass={cs.DS},
      url={https://arxiv.org/abs/2608.24866}, 
}
\bibliographystyle{alpha}

\appendix

\section{The Green--Ruzsa Construction}
\label{app:green-ruzsa}

We give the full proof of \Cref{thm:GR-black-box}.  This is the construction
of Green and Ruzsa~\cite[Section~5]{GreenRuzsa2019}, with the number of primes
$t$ left as a free parameter. 

Let $q_1<q_2<\cdots<q_t$ be the first distinct odd primes and denote the primorial
\[
        Q=\prod_{i=1}^t q_i.
\]
For every $d\in\{0,1,\ldots,Q-1\}$, let $x_d$ be the unique integer in $\{1,2,\ldots,Q\}$ such that
\[
        x_d\equiv d^2\pmod Q.
\]
Define
\[
        S
        :=\bigcup_{d=1}^{Q-1}
          \{x_d,x_d+d,\ldots,x_d+(K-1)d\}.
\]
For every $1\le d<Q$, the corresponding set in the union is a $K$-term arithmetic progression of difference $d$.  
Moreover,
\[
        1\le x_d+jd
        \le Q+(K-1)(Q-1)
        \le KQ,
\]
so $S\subseteq\{1,\ldots,KQ\}$.  
This proves the progression and diameter parts of the theorem.

It remains to bound the size of the union.  
For each $j\in\{0,1,\ldots,K-1\}$, define the $j$-th layer
\[
        S_j:=\{x_d+jd:0\le d<Q\}.
\]
Modulo $Q$, this layer is the image of
\[
        f_j(d)=d^2+jd.
\]

Fix one of the odd primes $q_i$.  
Since $2$ is invertible modulo $q_i$, we may complete the square:
\[
        d^2+jd
        =\left(d+\frac j2\right)^2-\frac{j^2}{4}
        \pmod{q_i}.
\]
Translation of the input and output does not change image size, and the square map on $\F_{q_i}$ has one value coming from $0$ and one value for each pair $\{z,-z\}$ of nonzero inputs.  
It therefore has exactly
\[
        \frac{q_i+1}{2}
\]
values.  
Hence the image of $f_j$ modulo $q_i$ has this same cardinality, independently of $j$.

The Chinese remainder theorem identifies $\Z/Q\Z$ with $\prod_i\F_{q_i}$, and the map $f_j$ acts coordinate-wise.
It follows that the image of $f_j$ modulo $Q$ has exactly
\begin{align*}
        \prod_{i=1}^t\frac{q_i+1}{2}
        &=Q2^{-t}\prod_{i=1}^t\left(1+\frac1{q_i}\right)\\
        &=Q2^{-t}\rho_t
\end{align*}
residues.  

Finally, $S_j\subseteq[1,KQ]$, and thus every residue modulo $Q$ has at most $K$ representatives in this interval.  
Combining this with the layer count gives
\[
        |S_j|\le KQ2^{-t}\rho_t.
\]
Summing over the $K$ layers yields
\[
        |S|
        \le\sum_{j=0}^{K-1}|S_j|
        \le K^2Q2^{-t}\rho_t,
\]
which is the required size bound.

\end{document}